\documentclass[amsthm]{arxelsart}
\usepackage{amssymb,url}
\usepackage{framed}

\usepackage{natbib}
\usepackage[dvipdfmx]{graphics}
\usepackage{fancybox,amsmath,amsbsy,graphicx} 

\usepackage{hyperref}

\newlength{\myheight}
\newlength{\myheighta}
\DeclareMathOperator{\res}{res}

\newcommand{\Set}[1]{\left\{#1\right\}}
\newcommand{\setDef}[2]{{#1}\left|\,\vphantom{#1}{#2}\right.}
\newcommand{\SetDef}[2]{\Set{\setDef{#1}{#2}}}

\usepackage{amssymb}

\journal{ \ }

\begin{document}

\begin{frontmatter}

\title{An effective method for computing Grothendieck point residue mappings }
\thanks{This work has been partly supported by JSPS  Grant-in-Aid for Scientific Research (C) (Nos 18K03214,18K03320).}

\author{Shinichi Tajima}
\address{Graduate School of Science and Technology, Niigata University, \\8050, Ikarashi 2-no-cho, Nishi-ku Niigata, Japan}
\ead{tajima@emeritus.niigata-u.ac.jp}

\author{Katsusuke Nabeshima}
\address{Graduate School of Technology, Industrial and Social Sciences, Tokushima University, Minamijosanjima-cho 2-1, Tokushima, Japan}
\ead{nabeshima@tokushima-u.ac.jp}

\begin{abstract}
Grothendieck point residue is considered in the context of computational complex analysis. A new effective method is proposed for computing Grothendieck point residues mappings and residues. 
Basic ideas of our approach are the use of Grothendieck local duality and a 
transformation law for local cohomology classes. A new tool is devised for efficiency to solve the extended ideal membership problems  in local rings. The resulting algorithms are described with an example to illustrate them. An extension of the proposed method to parametric cases is also discussed as an application.

\end{abstract}


\end{frontmatter}


\section{Introduction}
The theory of Grothendieck residue and duality is a cornerstone of algebraic geometry and complex analysis \citep{gh,gr,hartshorne,ku}. 
It has been used and applied 
in diverse problems of several different fields of mathematics \citep{bb2,bklp,cm,ds,g,le,ob,p,s2}. In the global situation, methods for computing the total sum of Grothendieck residues have been extensively studied and applied by several authors \citep{bklp,cds,ky,y}.

The concept of Grothendieck local residue together with the local duality theory also play quite important roles in complex analysis, especially in singularity theory \citep{bss,c,co,kl,ob,s}.  
Computing  \ Grothendieck local residues is therefore of fundamental importance. However, since the problem is local in nature, it is difficult in general to compute Grothendieck local residues \citep{ob2}. In fact, a direct use of the classical transformation law 
described in \citep{hartshorne} only gives algorithms which lack efficiency. Compared to the global situation, despite the importance, much less work has been done on algorithmic aspects of computing Grothendieck local residues \citep{em2,m,ot1,ot2,tn2}.
Grothendieck local residues with parameters are useful in the study of singularity theory, for example, deformations of singularity and 
unfoldings of holomorphic foliations \citep{kul,sai,v}. However, to the best of our knowledge,  existing algorithm of computing Grothendieck local residues are not designed to be able to treat parametric cases. 

In this paper, we consider methods for computing Grothendieck point residues
from the point of view of complex analysis and singularity theory. We propose a new effective
method for computing Grothendieck point residues mappings and residues, which can be extended to treat parametric cases.

Let $  X  \subset {\mathbb C}^n $ be an  open neighborhood of the origin $ O \in {\mathbb C}^n $ and let
 $ f_1(z),f_2(z),$ $\ldots,f_n(z)  $ be  $ n $ holomorphic functions defined on $ X $, 
where $ z=(z_1,z_2,\ldots,z_n) \in X$. 
  Assume that their common locus in $ X $ is the origin $ O $: 
$   \{ z \in X \mid f_1(z) = f_2(z)= \cdots = f_n(z) =0 \}  = \{ O \}$.

Then, for a given germ $ h(z) $ of holomorphic function at $ O $, the Grothendieck point residue at the origin $ O $, denoted by 
$$ {\rm res}_{\{O\}}\left(\frac{h(z)dz}{f_1(z)f_2(z) \cdots f_n(z)}\right), $$
of the differential form
$ \displaystyle{\frac{h(z)dz}{f_1(z)f_2(z) \cdots f_n(z)}} $ can be expressed, or defined, 
as the integral 
$$ 
\left(\frac{1}{2\pi \sqrt{-1}}\right)^n \int \cdots \int_{\gamma_{\epsilon}} \frac{h(z)dz}{f_1(z)f_2(z) \cdots f_n(z)}, 
$$
where $ dz= dz_1 \wedge dz_2 \wedge \cdots \wedge dz_n, $ and where  $ \gamma_{\epsilon} $ is a real $n$-dimensional cycle:
 $$ \gamma_{\epsilon} = \{ z \in X \mid |f_1(z)|=|f_2(z)|= \cdots =|f_n(z)| =\epsilon \},$$ 
with $ 0 < \epsilon \ll 1.  $
(See for instance, \citep{bb2,gh,ton}).

Let 
$$ h(z) \longrightarrow  {\rm res}_{\{O\}}\left(\frac{h(z)dz}{f_1(z)f_2(z) \cdots f_n(z)}\right) $$
be the Grothendieck point residue mapping that assigns to a holomorphic function $ h(z)$ the value of the Grothendieck point residue. We show that, based on the concept of local cohomology, the use of  Grothendieck local duality and a transformation law for local cohomology classes given by J. Lipman \citep{l} allows us to design an effective method for computing Grothendieck local residue mappings and another one for computing Grothendieck local residues. Note that the classical transformation law on Grothendieck residue is of no avail for computing Grothendieck local residue mappings. Since we compute Grothendieck local residue mappings, our method is applicable when the holomorphic function $h(z)$ in the numerator is computable, that is the case when the coefficients of the Taylor expansion of $h(z)$ is computable. This is an advantage of our approach. We also show that the proposed method can be extended to treat parametric cases.  This is another advantage of our approach.

In Section~2, we recall the transformation law for local cohomology classes and Grothendieck local duality.  In Section~3, we fix our notation and we briefly recall our basic tool, an algorithm for 
computing Grothendieck local duality. We devise, in the context of exact computation, a new tool which plays 
a key role in the resulting algorithm. In Section~4, we describe the resulting algorithm for computing Grothendieck point residue mappings and the algorithm for computing Grothendieck point residues. In Section~5, as an application, we generalize the proposed method to treat parametric cases and we show, by using an example, an algorithm for computing Grothendieck point residues associated to a $ \mu$-constant deformation of 
 quasi homogeneous isolated hypersurface singularities.

\section{Local analytic residues}

The concept of Grothendieck point residue was introduced by A. Grothendieck in terms of
derived categories and local cohomology. In this section, we briefly recall some basics on
transformation law for local cohomology classes and Grothendieck local duality.

Let $  X  \subset {\mathbb C}^n $ be an  open neighborhood of the origin $ O \in {\mathbb C}^n $.
Let $ {\mathcal O}_X $ be the sheaf on $ X $ of holomorphic functions, and $ \Omega_X^n $ 
the sheaf of holomorphic $ n $-forms. 
Let $ {\mathcal H}_{\{O\}}^n({\mathcal O}_X) $ (resp.  $ {\mathcal H}_{\{O\}}^n(\Omega_X^n) $ ) 
denote the local cohomology supported at $ O $ of $ {\mathcal O}_X $ (resp. $ \Omega_X^n $).

Then, $ {\mathcal O}_{X, O} $, the stalk at $ O $ of the sheaf $ {\mathcal O}_X $, and the local cohomology
$ {\mathcal H}_{\{O\}}^n(\Omega_X^n) $ are mutually dual as locally convex topological vector spaces \citep{bs}.
The duality is given by the point residue pairing:
$$ {\rm res}_{\{ O \}} ( \ast, \ \ast ): 
{\mathcal O}_{X, O} \times {\mathcal H}_{\{O\}}^n(\Omega_X^n)  \longrightarrow {\mathbb C} $$

Let $ F=[ f_1(z), f_2(z), \ldots, f_n(z)] $ be an $n$-tuple of $ n $ holomorphic functions defined on $X$. 
Assume that their common locus $ \{ z \in X \mid f_1(z) = f_2(z)= \cdots = f_n(z) =0 \} $ in $ X $ 
is the origin $ O. $ Let $ I_F $ denote the ideal in 
$ {\mathcal O}_{X,O} $ generated by $ f_1(z), f_2(z),\ldots, f_n(z). $ 
Let $ \omega_F $ denote a local cohomology class
$$ \omega_F = \left[ \begin{array}{c} dz \\ f_1(z) f_2(z) \cdots f_n(z) \end{array} \right]  $$
in  $ {\mathcal H}_{\{O\}}^n(\Omega_X^n) $, where
 $ dz=dz_1 \wedge dz_2 \wedge \cdots \wedge dz_n, $ and $ \left[ \quad \right] $ stands for 
Grothendieck symbol \citep{hartshorne,hartshorne1}. Residue theory says that, for $ h(z) $ in $ {\mathcal O}_{X,O}, $ one has 

$$ {\rm res}_{\{O\}}\left(\frac{h(z)dz}{f_1(z)f_2(z) \cdots f_n(z)}\right) = {\rm res}_{\{O\}}( h(z), \omega_F). $$

\subsection{Transformation law}

Since $ V(I_F)\cap X = \{ O \}, $ there exists, for each $ i=1,2,\ldots,n,$ a positive integer $ m_i $ such that $ z_i^{m_i} \in I_F. $ 
There exists an $n$-tuple of holomorphic functions $ a_{i, 1}(z), a_{i,2}(z),\ldots,a_{i,n}(z) $ such that  
$$ z_i^{m_i} = a_{i,1}(z)f_1(z) + a_{i,2}(z)f_2(z) + \cdots + a_{i,n}(z)f_n(z), \qquad i=1,2,\ldots,n. $$
Set $ A(z) = \det (a_{i,j}(z))_{1\leq i,j \leq n}. $

We have the following key lemma \citep{l}.

\begin{lem}[Transformation law for local cohomology classes] \ \\
In  $ {\mathcal H}_{\{O\}}^n(\Omega_X^n), $  the following formula holds.
$$ \omega_F =  
\left[ \begin{array}{c} A(z) dz \\ z_1^{m_1} z_2^{m_2} \cdots z_n^{m_n} \end{array} \right].  $$

\end{lem}

For the proof of the result above, we refer the reader to \citep{ku,l}. 
Note that the formula above implies the classical transformation law 
$$  {\rm res}_{\{O\}}\left(\frac{h(z)dz}{f_1(z)f_2(z) \cdots f_n(z)}\right)= 
{\rm res}_{\{O\}}\left(\frac{h(z)A(z)dz}{z_1^{m_1} z_2^{m_2} \cdots z_n^{m_n}}\right) $$
for point residues described in \citep{hartshorne}. See also 
\citep{bb2,bh,gh,ky}.

\subsection{Grothendieck local duality}

We define $ W_F $ to be the set of local cohomology classes in $  {\mathcal H}_{\{O\}}^n(\Omega_X^n) $ 
that are killed by $ I_F: $
$$ W_F = \{ \omega \in  {\mathcal H}_{\{O\}}^n(\Omega_X^n)  \mid 
f_1(z)\omega = f_2(z)\omega = \cdots = f_n(z)\omega = 0 \}. $$
Then, according to Grothendieck local duality, the pairing 
$$ {\rm res}_{\{O\}}( *, \ *): {\mathcal O}_{X,O}/I_F \times W_F \longrightarrow {\mathbb C} $$
induced by the residue mapping is non-degenerate \citep{ak,gr,hartshorne,l2}.

Let $ \succ^{-1} $ be a local term ordering on the local ring $ {\mathcal O}_{X,O} $ and let 
$ \{ z^{\alpha} \mid \alpha \in \Lambda_F \} $ denote the monomial basis of the quotient space 
$ {\mathcal O}_{X,O}/I_F $ with respect to the local term ordering $ \succ^{-1}, $ where
$ \Lambda_F  \subset {\mathbb N}^n $ is the set of exponents $ \alpha $ of basis monomials $ z^{\alpha} $.

Let $ \{ \omega_{\alpha}  \in W_F    \mid \alpha \in \Lambda_F \} $ denote the 
dual basis of 
$ \{ z^{\alpha} \mid \alpha \in \Lambda_F \} $ with respect to the Grothendieck point residue.  Then, we have 
\begin{enumerate}
\item[(i)] $ {\mathcal O}_{X,O}/I_F \cong {\rm Span}_{{\mathbb C}}\{ z^{\alpha} \mid \alpha \in \Lambda_F \}, $ \\

\item[(ii)] $ W_F = {\rm Span}_{{\mathbb C}}\{ \omega_{\alpha} \mid \alpha \in 
\Lambda_F \}, $  \\

\item[(iii)] $ {\rm res}_{\{O\}}( z^{\alpha}, \omega_{\beta}) = 
\left\{ \begin{array}{rl} 1, \quad & \alpha = \beta, \\
                                        0, \quad & \alpha \ne \beta, \qquad \alpha, \beta \in \Lambda_F. \\
         \end{array} \right. $
\end{enumerate}

\subsection{Residue mapping}

Since $ \omega_F $ satisfies $ f_1(z)\omega_F = f_2(z)\omega_F = \cdots = f_n(z)\omega_F = 0, $ 
the local cohomology class $ \omega_F $ is in $ W_F $. Therefore $ \omega_F $ can be expressed as a 
linear combination of the basis $ \{ \omega_{\alpha} \mid  \alpha \in \Lambda_F \}. $

Assume that, for the moment, we have the following expression:
$$ \omega_F = \sum_{\alpha \in \Lambda_F} b_{\alpha}\omega_{\alpha}, \quad b_{\alpha} \in {\mathbb C}. $$

Now let 
$$ {\rm NF}_{\succ^{-1}}(h)(z) = \sum_{\alpha \in \Lambda_F} h_{\alpha}z^{\alpha}, 
\quad h_{\alpha} \in {\mathbb C} $$
be the normal form of the given holomorphic function $ h(z). $  Then, we have the following.

\begin{thm}

$$ {\rm res}_{\{O\}}\left(\frac{h(z)dz}{f_1(z)f_2(z) \cdots f_n(z)}\right)
= \sum_{\alpha \in \Lambda_F} h_{\alpha} b_{\alpha} $$

\end{thm}
\begin{proof}
Since $ h - {\rm NF}_{\succ^{-1}}(h) \in I_F, $ we have
\begin{center}
$ {\rm res}_{\{O\}}( h(z), \omega_F) =  {\rm res}_{\{O\}}( {\rm NF}_{\succ^{-1}}(h)(z) , \omega_F). $
\end{center}
Therefore, 
\begin{center}
$\displaystyle 
 {\rm res}_{\{O\}}\left(\frac{h(z)dz}{f_1(z)f_2(z) \cdots f_n(z)}\right) =
{\rm res}_{\{O\}}\left( \sum_{\alpha \in \Lambda_F}h_{\alpha}z^{\alpha}, \sum_{\beta \in \Lambda_F}b_{\beta}\omega_{\beta}\right), $
\end{center}
which is equal to 
\begin{center}
$
\displaystyle \sum_{\alpha, \beta \in \Lambda_F} h_{\alpha}b_{\beta}{\rm res}_{\{O\}}(z^{\alpha}, \omega_{\beta})
=\sum_{\alpha \in \Lambda_F} h_{\alpha} b_{\alpha}.
$
\end{center}
This completes the proof.
\end{proof}

\section{Tools}
Let us  consider a method for computing Grothendieck point residues in the context of symbolic 
computation. We start by recalling some basics on an algorithm for computing Grothendieck local duality given 
in \citep{tn3,tnn}.

Let $ K={\mathbb Q} $ be the field of rational numbers and let $ z=(z_1,z_2,\ldots,z_n) \in {\mathbb C}^n. $
Let $H^n_{[O]}(K[z])$ denote the algebraic local cohomology defined to be 
$$ H^n_{[O]}(K[z])= \lim_{k \rightarrow \infty }Ext^n_{K[z]}(K[z]/{\mathfrak m}^k, \Omega_X^n), $$ 
where $ {\mathfrak m} $ is the maximal ideal $ {\mathfrak m} = \langle z_1, z_2,\ldots,z_n \rangle $ in $ K[z] =K[z_1,z_2,\ldots,z_n].$

We adopt the notation used in \citep{nt2,nt3,nt5,nt4} to handle local cohomology classes. 
For instance, a polynomial $ \sum_{\lambda}c_{\lambda} \xi^{\lambda} $ in
 $ K[\xi] = K[\xi_1,\xi_2,\ldots,\xi_n] $ represents 
the local cohomology class of the form
$\displaystyle  \sum_{\lambda=(\ell_1, \ell_2,\ldots,\ell_n)} c_{\lambda} 
\left[ \begin{array}{c} 1 \\ z_1^{\ell_1+1} z_2^{\ell_2+1} \cdots z_n^{\ell_n+1} \end{array} \right]. $
Note that a multiplication on $ \xi^{\beta} $ by $ z^{\alpha} $ is 
$$ 
z^{\alpha} \ast \xi^{\beta} = 
\left\{
\begin{array}{cl} \xi^{\beta-\alpha},  & \beta \geq \alpha. \\
                                                                                    0, & \text{otherwise.} \\
\end{array} \right. 
$$
Let $ \succ $ be a term ordering on $ K[\xi] $. 
For a local cohomology class
$\displaystyle  \psi=c_{\alpha} \xi^{\alpha} + \sum_{\xi^\alpha \succ \xi^\gamma} c_{\gamma}\xi^{\gamma} $, 
we call 
$ \xi^{\alpha} $ the head monomial of $ \psi $, and $ \alpha \in {\mathbb N}^n$ the head exponent of $ \psi $.

Let $ F = [f_1(z), f_2(z), \ldots,f_n(z)] $ be a list of $n$ polynomials  $f_1,\dots,f_n$ in $K[z] $.
We also assume as in the previous section that there exists an open neighborhood $ X $ of the origin $ O $ such that  
their common locus is the origin: $ \{ z \in X \mid f_1(z)=f_2(z)= \cdots =f_n(z)=0 \} = \{O\}. $

We set $$ H_F = \{ \psi \in  H^n_{[O]}(K[z]) \mid
 f_1(z) \ast \psi = f_2(z) \ast \psi = \cdots = f_n(z) \ast \psi = 0 \}. $$

\subsection{Algorithm for computing Grothendieck local duality}

In \citep{nt6,tnn}, an algorithm for computing bases of $ H_F $ is introduced.  Let $ \Psi_F $ denote an  output of the
algorithm. Then, 
$$ W_F = {\rm Span}_{{\mathbb C}}\{ \psi dz \mid \psi \in \Psi_F \} $$
holds. Furthermore, the algorithm computes Grothendieck local duality 
with respect to the Grothendieck local residue pairing. Here we recall some basic properties of the algorithm. 

An output of the algorithm, say $ \Psi_F, $ a basis of the vector space $ H_F$, has the following form:
$$ \Psi_F = \SetDef{ \psi_{\alpha}}{ \psi_{\alpha} 
= \xi^{\alpha} + \sum_{\xi^\alpha \succ \xi^\gamma} c_{\gamma}\xi^{\gamma}, \quad \alpha \in \Lambda_F},
$$
where $ \Lambda_F \subset {\mathbb N}^n $ is the set of the head exponents of local cohomology classes 
in $ \Psi_F$.

Let $ L_F $ denote the set of lower exponents of local cohomology classes in $ \Psi_F: $
$$ L_F = \SetDef{ \gamma \in {\mathbb N}^n }{ \exists \
 \psi_{\alpha} = \xi^{\alpha} + \sum_{\xi^\alpha \succ \xi^\gamma} c_{\gamma}\xi^{\gamma} \in \Psi_F  \text{ such that } \ c_{\gamma} \ne 0 }. $$

Set $ E_F = \Lambda_F \cup L_F $ and $ T_F = \{ \xi^{\lambda} \mid \lambda \in E_F \}. $
Now let $ \ell_{F, i} = {\rm max} \{ \ell \mid \xi_i^{\ell} \in T_F \}. $ 
Then Grothendieck local duality implies the following.

\begin{lem}
Set $ m_i = \ell_{F, i} + 1. $ Then $ z_i^{m_i} \in I_F $ holds, where 
$ I_F $ is 
the ideal in the local ring $ K\{z\} $ generated by $ f_1(z), f_2(z),\ldots , f_n(z). $
\end{lem}
\begin{proof} 
Since $ z_i^{m_i} \ast \psi_{\alpha} = 0 $ and $  \psi_{\alpha} \in \Psi_F $ hold, we have
$ z_i^{m_i} \ast \psi = 0 $ for $  \psi \in H_F. $ It follows from the Grothendieck local duality that $ z_i^{m_i}$ is in $I_F. $
\end{proof}

Now let us consider the set of monomials $ M_F $ in $ K\{ z \} $ defined to be
$ M_F = \{ z^{\alpha} \mid \alpha \in \Lambda_F \} $. Let $ \succ^{-1} $ denote the 
local term ordering on $ K\{z\} $ defined as the inverse ordering of $ \succ. $
Then, $ M_F $ constitutes a monomial basis of the quotient $ K\{z\}/I_F $ with respect to the local term 
ordering $ \succ^{-1}. $ Furthermore, we have the following result \citep{tn1,tn2,tn3}.

\begin{thm}
Let $ \Psi_F, M_F $ be as above. Then, $ \Psi_F $ is the dual basis of the basis $ M_F $ with respect to Grothendieck 
local residue pairing. That is, for $ z^{\alpha} \in M_F $ and for $ \psi_{\beta} \in \Psi_F,$ 
$$ {\rm res}_{\{O\}}( z^{\alpha}, \psi_{\beta}dz) = 
\left\{ \begin{array}{rl} 1, \quad & \alpha = \beta, \\
                                        0, \quad & \alpha \ne \beta, \\
         \end{array} \right. $$
holds.
\end{thm}
\noindent 
{\it Sketch of the proof}. \ 
Since the algorithm outputs a {\it reduced} basis of $ H_F, $ we have $ \Lambda_F \cap L_F = \emptyset, $
which implies the result. $\hfill \qed$

\subsection{A key tool}\label{sub32}
Let $ m_i $ be an integer such that $ z_i^{m_i} $ is in the ideal $ I_F = (f_1, f_2,\ldots,f_n) $ in the local 
ring. Then there exist germs $ a_{i,1}(z), a_{i,2}(z),\ldots, a_{i,n}(z) $ of holomorphic functions such that 
$$ z_i^{m_i} = a_{i,1}(z)f_1(z) + a_{i,2}(z)f_2(z) + \cdots + a_{i,n}(z)f_n(z), \qquad i=1,2,\ldots,n. $$
Theory of symbolic computation asserts that such $ n$-tuple of holomorphic functions can be obtained by 
computing syzygies in the local ring $ K\{z\}. $ Whereas, since the cost of computation of syzygies in local rings
is high, a direct use of the classical algorithm of computing syzygy is not appropriate  in actual computations. In fact, it is difficult to obtain these holomorphic functions. In previous papers \citep{nt4}, the authors of the present paper have proposed a new effective method to overcome this type of difficulty. 

We adopt the proposed method mentioned above and devise a new, much more  efficient algorithm by 
improving the previous algorithm presented in \citep{nt3,nt4}.
We start by recalling the main idea given in \citep{nt4}.
Let $ J_F=\langle f_1(z), f_2(z),\ldots,$ $f_n(z) \rangle $ denote the ideal in the polynomial ring $ K[z] $ generated
 by $ f_1(z), f_2(z),$ $\ldots, f_n(z). $ Let $ J_{F,O} $ be the primary component of $ J_F $ whose associated prime 
is the maximal ideal $ {\mathfrak m}=\left< z_1, z_2,\ldots, z_n \right>, $ and $ G_Q $ a Gr\"obner basis of the ideal quotient
$ Q=J_F: J_{F,O} \subset K[z]. $ Then there is in $ G_Q $ a polynomial, say $ q(z), $ such that $ q(O) \ne 0. $

Now let $ r(z) \in J_{F, O}. $ Then, since $ q(z)r(z) \in J_F, $ there exists an $n$-tuple of polynomials 
$ p_1(z), p_2(z),\ldots,p_n(z) $ in $ K[z] $, such that  
$$ q(z)r(z) = p_1(z)f_1(z) + p_2(z)f_2(z) + \cdots + p_n(z)f_n(z). $$

Since, $ q(O) \ne 0, $ we have a following  expression in the local ring $ K\{z\}: $
$$ r(z) = \frac{p_1(z)}{q(z)}f_1(z) + \frac{p_2(z)}{q(z)}f_2(z) + \cdots + \frac{p_n(z)}{q(z)}f_n(z). $$

Since $ I_F = K\{z\} \otimes J_{F,O} $ and $ z_i^{m_i} \in I_F, $ $ z_i^{m_i} \in J_{F,O} $ holds.
Therefore, the argument above can be applied to compute 
 germs $ a_{i,1}(z), a_{i,2}(z),\ldots , a_{i,n}(z) $ of holomorphic functions.
 Note also that, since $ J_{F, O} = \{ p(z) \in K[z] \mid p(z)*\psi_{\alpha} = 0, \ \psi_{\alpha} \in \Psi_F \}, $
 the primary ideal $ J_{F,O} $ can be computed by using $ \Psi_F. $

Let $ G_F =\{ g_1, g_2,\ldots, g_{\nu} \} $ be a Gr\"obner basis of $ J_F. $ 
Let $ R_F $ be a list of relations between $ g_j $ and $ F=[f_1,f_2,\ldots,f_n]: $
\begin{center}
$ g_j = r_{1,j}f_1 + r_{2,j}f_2 + \cdots + r_{n,j} f_n, $
\end{center}
where $ r_{i,j} \in K[z],  \ i=1,2,\ldots,n$, and $j=1,2,\ldots, \nu. $
Let $ S_F $ be a Gr\"obner basis of the module of syzygies among $ F: $
$$ s_1f_1+s_2f_2 + \cdots + s_nf_n = 0, $$
where $ s_i \in K[z], \ i=1,2,\ldots,n. $ Let $ q $ be a polynomial in $ G_Q $ such that $ q(O) \ne 0. $

Now we are ready to present a new tool.

\noindent
\hrulefill\\
{\bf Algorithm 1}. {\bf localexpression} \vspace{-3.0mm} \\
 \mbox{}\hrulefill  \\
{\bf Input}: $ G_F, R_F, S_F, q, r. $ \\
{\bf Output}: $ [p_1, p_2,\ldots,p_n] $ \ such that  
$ q(z)r(z) =p_1(z)f_1(z) + p_2(z)f_2(z)+ \cdots + p_n(z)f_n(z). $ \\
{\bf BEGIN}
\begin{enumerate}
\setlength{\leftskip}{7mm}
\item[step 1:] divide $ qr$ by the Gr\"obner basis $ G_F=\{ g_1, g_2,...,g_{\nu} \} $:
\begin{center}
$ qr= e_1g_1+e_2g_2 + \cdots + e_{\nu}g_{\nu}; $
\end{center}
\item[step 2:]
rewrite the relation above by using $ R_F $:
\begin{center} 
$ qr=\left(\sum_{j}r_{j,1} e_j\right)f_1+\left(\sum_{j}r_{j,2}e_j\right)f_2 + \cdots + \left(\sum_{j}r_{j,2}e_j\right)f_n; $
\end{center}
\item[step 3:] 
simplify the expression above by using $ S_F$:
$$ q(z)r(z) =p_1(z)f_1(z) + p_2(z)f_2(z)+ \cdots + p_n(z)f_n(z); $$
\end{enumerate}
\noindent
return $ [p_1,p_2,\ldots ,p_n] ; $\\
{\bf END} \vspace{-3mm}\\
\noindent 
 \mbox{}\hrulefill \vspace{-3mm} \\
\noindent 

\begin{exmp}[$ E_{12} $ singularity] 
Let $ f(x, y) = x^3+y^7+xy^5 $ and let $ F=[\frac{\partial f}{\partial x}(x,y), $ $\frac{\partial f}{\partial y}(x,y) ]. $ Note that $ f(x,y) $ is a semi quasi-homogeneous function with respect to the weight vector $ (7,3). $ Let $ \succ $ be the weighted degree lexicographical ordering on $ K[\xi, \eta] $ with respect to the weight vector $ (7,3), $ where $ \xi, \eta $ correspond to $ x, y. $

Then, $ \dim_K(H_F) = 12, $ the Milnor number at the origin $ (0, 0 ) $ of the curve $ \{ (x,y) \in {\mathbb C}^2 \mid f(x,y) = 0 \}. $  The algorithm for computing Grothendieck local duality, mentioned in this 
section, outputs a basis $ \Psi_F $ that consists of the following 12 local cohomology classes;
\begin{flushleft}
$
1, \eta, \xi, \eta^2, \xi\eta, \eta^3, \xi\eta^2, \eta^4, \eta^5-\frac{1}{3}\xi^2, \xi\eta^3, \xi\eta^4-\frac{5}{7}\eta^6+\frac{5}{21}\xi^2\eta,$\\
$\xi\eta^5-\frac{5}{7}\eta^7 -\frac{1}{3}\xi^3+\frac{5}{21}\xi^2\eta^2.
$
\end{flushleft}

Note for instance that the local cohomology class $  \left[ \begin{array}{c} 1 \\ x y^6  \end{array} \right]  -\frac{1}{3}\left[ \begin{array}{c} 1 \\ x^3 y \end{array} \right]  $ represented by $ \psi_{(0,5)}=\eta^5-\frac{1}{3}\xi^2 $ above acts on a holomorphic function $ h(x,y) = \sum_{(i,j)}c_{(i,j)}x^i y^j $ by 
\begin{center}
$ {\rm res}_{O}(h(x,y), \psi_{(0,5)}dx\wedge dy) = c_{(0,5)} -\frac{1}{3}c_{(2,0)}. $
\end{center}

The output implies that 
\begin{flushleft}
$ \Lambda_F = \{ (0,0), (0,1),(0,2), (1,0),(0.3), (1,1),  (0,4), (1,2), (0,5), (1,3), (1,4), $\\
$(1,5) \} $
\end{flushleft} 
and $ M_F = \{ x^i y^j \mid (i,j) \in \Lambda_F \} $ is the 
monomial basis of the quotient space $ K\{x,y\}/I_F $ with respect to the local term ordering $ \succ^{-1} $
on $ K\{x,y\}, $ where $ I_F $ denote the ideal in $ K\{ x,y \} $ 
generated by $  \frac{\partial f}{\partial x}, \frac{\partial f}{\partial y} $. 
Furthermore $ W_F =\{ \psi dx \wedge dy \mid \psi \in \Psi_F \} $ is the dual basis of the monomial basis 
$ M_F $ with respect to the Grothendieck local residue pairing. 
Since $ \lambda_F =(3,7) $, we have $ x^4, y^8 \in I_F. $ 
 
Let $ J_F $ be the ideal in $ K[x,y] $ generated by the two polynomials $ \frac{\partial f}{\partial x}, \frac{\partial f}{\partial y}. $ Let $ J_{F, O} $ be the primary component of $ J_F $ whose associated prime is the maximal ideal $ \left< x, y \right>. $ A Gr\"obner basis of the ideal quotient $ J_{F,O} : J_F $ is\begin{center} 
$ 3125x+151263, \ 25y+147. $
\end{center}
Set $ q(x,y) = 25y+147. $ Then, the algorithm {\bf localexpression} outputs the following: \vspace{-1cm}
\begin{eqnarray*}
 q(x,y)x^4 &=& (49x^2+ 25/3x^2y-49/3y^5) \frac{\partial f}{\partial x} 
                        + ( -5/3xy^2+7/3y^4) \frac{\partial f}{\partial y}, \\
 q(x,y)y^8 &=& 25y^4 \frac{\partial f}{\partial x} + (-15x+21y^2)\frac{\partial f}{\partial y}.
\end{eqnarray*}
\  \vspace{-1.3cm}
\end{exmp}

\section{Algorithms}
Let $ \tau_F $ denote the local cohomology class in $ H_F $ defined to be
$$ \tau_F = \left[ \begin{array}{c} 1 \\ f_1(z) f_2(z) \cdots f_n(z) \end{array} \right].  $$
Since $ \omega_F = \tau_F dz,  $ the local cohomology class $ \tau_F $ is the kernel function of the 
point residue mapping. 

Let
\begin{center}
$ q(z)z_i^{m_i} = p_{i, 1}(z)f_1(z) + p_{i,2}(z)f_2(z) + \cdots + p_{i,n}(z)f_n(z), \ i=1,2,\ldots,n, $
\end{center}
and set $ {\rm Det}(z) = {\rm det}(p_{i,j}(z))_{1 \leq i,j \leq n}.  $

Let $ I_M $ be the ideal in $ K[z] $ generated by $ z_1^{m_1}, z_2^{m_2},\ldots,z_n^{m_n}. $
Let $ u(z) \in K[z] $ be a polynomial such that  
$ u(z)q(z) -1 \in I_M $. 

Since $  A(z)=\det (p_{i,j}(z)/q(z))_{1\leq i,j \leq n} $ is equal to $ \frac{1}{q(z)^n}{\rm Det}(z), $
the transformation law 
implies the following
\begin{center}
$\displaystyle  \tau_F = 
\left[ \begin{array}{c} u(z)^n {\rm Det}(z)  \\ z_1^{m_1} z_2^{m_2} \cdots z_n^{m_n} \end{array} \right].  $
\end{center}

Let $ \lambda_F =(\ell_{F,1}, \ell_{F,2},\ldots, \ell_{F,n}). $ Since $ m_i = \ell_{F,i} +1, $ 
the formula above can be rewritten as 
$ \tau_F= u(z)^n{\rm Det}(z) \ast \xi^{\lambda_F}. $

Note that, according to an algorithm in \citep{ss} discovered by Y. Sato and A. Suzuki, the inverse $ u(z) $ of $ q(z) $ in $ K[z]/I_M $ can be obtained  by using Gr\"obner basis computation.

The following algorithm computes a representation of the local cohomology class 
$ \tau_F $, the kernel function of the point residue mapping.

\noindent
\hrulefill\\
{\bf Algorithm 2}. {\bf tau} \vspace{-2.0mm} \\
 \mbox{}\hrulefill \\ 
{\bf Input:} $ V=[z_1,z_2,\ldots,z_n], \succ,  F = [f_1(z), f_2(z),\ldots, f_n(z) ]. $\\
 \hspace{12ex}  /* $V$: a list of variables, $\succ$: a term order */ \\
{\bf Output:} $ \tau_F= \sum_{\alpha \in \Lambda_F} b_{\alpha}\psi_{\alpha}.$\\
{\bf BEGIN}
\begin{enumerate}
\setlength{\leftskip}{8mm}
\item[step 1:] compute a basis $ \Psi_F = \{ \psi_{\alpha} \mid \alpha \in \Lambda_F \} $ of the  space $ H_F $; \\
\hspace{5ex}   /* $ \Lambda_F $: the set of head terms of $ \Psi_F $  */
\item[step 2:] compute $ \ell_{F,i} = \max \{ \ell | \xi_i^{\ell} \in T_F \} $ and set $ m_i = \ell_{F,i} +1, \ i=1,2,\ldots,n$; \\
\hspace{5ex} /* $ T_F = \{ \xi^{\lambda} \mid \lambda \in E_F \} $ */
\item[step 3:] compute a Gr\"obner basis of the ideal 
    $$ J_{F,O} = \{ p(z) \in K[z] \mid p(z) \ast \psi_{\alpha} = 0, \alpha \in \Lambda_F \}; $$

\item[step 4:] compute $ G_F, R_F, S_F$; \\
\hspace{5ex} /* notations are from subsection~\ref{sub32}*/

\item[step 5:] compute a Gr\"obner basis $ G_Q $ of the quotient ideal $ Q= J_F : J_{F,O} $ and 
choose a polynomial $ q(z) $ from $ G_Q $ such that $ q(O) \ne 0$;

\item[step 6:] compute
$$ q(z)z_i^{m_i} =  p_{i, 1}(z)f_1(z) + p_{i,2}(z)f_2(z) + \cdots + p_{i,n}(z)f_n(z), \ (i=1,2,\ldots,n), $$
by using the algorithm {\bf localexpression}; 

\item[step 7:] compute $ {\rm Det}(z) = {\rm det}(p_{i,j}(z))_{1 \leq i,j \leq n} $ and set $ {\rm ND}= {\rm NF}_{I_M}({\rm Det}(z) ),  $ the normal form of $ {\rm Det}(z) $ with respect to $ I_M $;

\item[step 8:] compute a Gr\"obner basis of the ideal in $ K[z, u] $ generated by
 $$ 1-q(z)u, z_1^{m_1}, z_2^{m_2},\ldots, z_n^{m_n} $$ 
with respect to an elimination ordering to eliminate $ u $;
\item[step 9:] choose a polynomial of degree one with respect to $ u $, of the  form $ cu + poly(z) $, from the Gr\"obner basis of step~8 and set 
$$ 
{\rm Den}= (-c)^n, \ {\rm NU}= {\rm NF}_{I_M}(poly(z)^n), \ {\rm Num}= {\rm NF}_{I_M}({\rm ND} \times {\rm NU}); $$

\item[step10:] compute $ \psi = {\rm Num} \ast \xi^{\lambda_F} $ and  
set $ {\rm Coeff}=\{ c_{\alpha} \mid \alpha \in \Lambda_F \} $; \\
\hspace{5ex} /* $ c_{\alpha} $ is the coefficient of a term $ \xi^{\alpha} $ of $ \psi $, $ \alpha \in \Lambda_F. $ */ 
\end{enumerate}
return  $ [ \Lambda_F, \Psi_F, {\rm Coeff}, {\rm Den} ]; $\\
{\bf END}\vspace{-3mm}\\
\noindent 
 \mbox{}\hrulefill \vspace{-3mm} \\
\noindent

The return of the algorithm above means 
$$ \tau_F = \frac{1}{{\rm Den}}\sum_{\alpha \in \Lambda_F} c_{\alpha} \psi_{\alpha}. $$
Note that, since, 
$$ {\res}_{O}(h(z)\tau_F dz) = \frac{1}{{\rm Den}}\sum_{\alpha \in \Lambda_F}b_{\alpha} {|res}_{O}(h(z)\psi_{\alpha}dz) $$
holds, the output of the algorithm above completely describes the Grothendieck point residue mapping
$$ h(z) \longrightarrow  {\rm res}_{\{O\}}\left(\frac{h(z)dz}{f_1(z)f_2(z) \cdots f_n(z)}\right). $$

\vspace{2ex}

Let $ {\rm Res}_F = {\bf tau}(V, \succ, F) $ be the output of the algorithm  {\bf tau}.
The following algorithm {\bf residues} evaluates the value of Grothendieck point residue. 

\noindent
\hrulefill\\
{\bf Algorithm 3}. {\bf residues} \vspace{-2.0mm} \\
 \mbox{}\hrulefill\\
{\bf Input:} $ h \in K[z], $ \ $ {\rm Res}_F. $\\
{\bf Output:} $ {\rm res}_{\{O\}}(h(z)\tau_F dz). $\\
{\bf BEGIN}
\begin{enumerate}
\setlength{\leftskip}{7mm}
\item[step 1:] compute the normal form of $ h$ by using $ \Psi_F $, i.e.,  $\displaystyle {\rm NF}_{\succ}(h)(z) = \sum_{\alpha \in \Lambda_F} h_{\alpha} z^{\alpha} $;

\item[step 2:] compute $\displaystyle  {\rm sum} = \sum_{\alpha \in \Lambda_F} h_{\alpha}c_{\alpha} ;$
\end{enumerate}
return $\displaystyle  \frac{{\rm sum}}{{\rm Den}} $ ;\\
{\bf END}\vspace{-3mm}\\
\noindent 
 \mbox{}\hrulefill \vspace{-3mm} \\
\noindent

Note that $ {\rm NF}_{\succ}(h) $ is  computed by the algorithms given in \citep{tn3,tnn}. The algorithm is free from standard bases computation. 
All the algorithms given in the present paper are implemented in a computer algebra system Risa/Asir{\citep{noro}).

\begin{exmp}[$ E_{12} $ singularity]
Let us continue the computation. Since step 1 to step~6 are done, we start from step 7.
From
$$\left(
\begin{array}{cc}
p_{1,1}& p_{1,2}\\
p_{2,1}& p_{2,2} 
\end{array} 
\right)=
\left(
\begin{array}{cc}
25/3x^2y+49x^2& -5/3xy^2+7/3y^4\\
25y^4& -15x+21y^2 
\end{array} 
\right),
$$
we have the determinant 
$$
{\rm Det} =
  (-125y-735)x^3+(175y^3+1029y^2)x^2+(125/3y^6+245y^5)x-175/3y^8-343y^7. 
$$

A Gr\"obner basis of the ideal in $ K[x,y,u] $ generated by $ 1-uq(x,y), x^4, y^8 $ with respect to a elimination ordering  $ u \succ x, y $ is \vspace{3mm}\\
\noindent
$\{x^4, y^8, -6103515625y^7+35888671875y^6-211025390625y^5+1240829296875y^4$ 
$-7296076265625y^3 +42900928441875y^2-252257459238225y
-21804125746715$ 
$2161u+1483273860320763 \}.$ \vspace{3mm}

We have 


\noindent
$ {\rm Num} = (6654091109227055694580078125y^7-391260557222550874841308593$ 
$75y^6+230061207646859914406689453125y^5-1352759900963536296711333984$ 
$375y^4+7954228217665593424662643828125y^3-46770861919873689337016345$ 
$709375y^2+275012668088857293301656112771125y-1617074488362480884613$ 
$737943094215)x^3+(-322085690705603880169365234375y^7+189386386134895$ 
$0815395867578125y^6-11135919504731830794527701359375y^5+6547920668$
$7823165071822883993125y^4 -385017735324400210622318557879575y^3+22639$ 
$04283707473238459233120331901y^2)x^2+(1559028730662456311233878190312$ 
$5y^7-91670889362952431100552037590375y^6+539024829454160294871245981$
$031405y^5)x-754634761235824412819744373443967y^7 $ \vspace{3mm} \\
and 
$ {\rm Den} = (218041257467152161)^2. $

Since $\displaystyle  b_{\alpha} = \frac{c_{\alpha}}{{\rm Den}}, $ we have 
$ \tau_F = \displaystyle{\frac{1}{{\rm Den}}({\rm Num}\ast(\xi^3\eta^7))}. $ 

\noindent
Therefore,

\vspace{1ex}
\noindent
$ \tau_F = 
30517578125/218041257467152161-1220703125/1483273860320763\eta+4$ 
$8828125/10090298369529\eta^2-1953125/68641485507\eta^2+78125/466948881\eta^4-3$ 
$125/3176523\eta^5+125/21609\eta^6-5/147\eta^7-9765625/1441471195647\xi+390625/9$ 
$805926501\xi\eta-15625/66706983\xi\eta^2+625/453789\xi\eta^3-25/3087\xi\eta^4+1/21\xi\eta^5+3125/9529569\xi^2-125/64827\xi^2\eta+5/441\xi^2\eta^2-1/63\xi^3.$\vspace{3mm}

This yields 
$$ \tau_F= \sum_{0 \leq i, j \leq 5} b_{i,j}\psi_{i,j}, $$
where \ 
$ b_{0,0} = 30517578125/218041257467152161,$ \\
$b_{0,1}=-1220703125/1483273860320763, \\ b_{0,2}=48828125/10090298369529, b_{0,3} = -1953125/68641485507, \\
b_{0,4} = 78125/466948881, b_{0,5} = -3125/3176523, \\
b_{1,0} = -9765625/1441471195647,
b_{1,1} = 390625/9805926501, \\
b_{1,2}=-15625/66706983, 
b_{1,3}=625/453789, 
b_{1,4}=-25/3087, b_{1,5}=1/21.$ \vspace{3mm} \\
\noindent 
and 

\noindent
$ \psi_{0,0}=1, \ \psi_{0,1}=\eta, \ \psi_{0,2}=\eta^2, \ \psi_{0,3}=\eta^3,\  \psi_{0,4}=\eta^4, \ \psi_{0,5}=\eta^5-\dfrac{1}{3}\xi^2, \\
\psi_{1,0}=\xi,\  \psi_{1,1}=\xi\eta, 
\psi_{1,2}=\xi\eta^2, \ \psi_{1,3}=\xi\eta^3, \\ 
\psi_{1,4}=\xi\eta^4-\dfrac{5}{7}\eta^6 +\dfrac{5}{21}\xi^2\eta, \ 
\psi_{1,5} = \xi\eta^5-\dfrac{5}{7}\eta^7 - \dfrac{1}{3}\xi^3 + \dfrac{5}{21}\xi^2\eta^2. $

\vspace{2ex}
Let $ {\rm NF}_{\succ}(h)(x,y) = \sum_{(i,j) \in \Lambda_F} h_{i,j}x^i y^j. $
Then, 
$$ {\rm res}_{\{O\}}(h(x,y), \tau_F dx \wedge dy) = \sum_{(i,j) \in \Lambda_F} h_{i,j}b_{i,j}. $$

\noindent
We have for instance, \ 
$ \displaystyle{ {\rm res}_{\{O\}}\left(\frac{dx \wedge dy}{\frac{\partial f}{\partial x} \frac{\partial f}{\partial y}}\right)
= \frac{30517578125}{218041257467152161}.} $

Recall that , as
local cohomology class $ \omega_F = \tau_F dx \wedge dy $ is in $ {\mathcal H}_{\{(0,0)\}}^2( \Omega_X^{2}),  $ the cohomology class $ \tau_F $ defines the residue mapping
$$ {\rm res}_{\{ O \}} ( \ast, \tau_F ): 
{\mathcal O}_{X, O}   \longrightarrow {\mathbb C}. $$

\noindent
Therefore, the formula above is valid for {\it germs of holomorphic functions} $ h(x,y) $. More precisely, 
for a germ of holomorphic function $\displaystyle  h(x,y) = \sum_{(i,j)} c_{i,j}x^i y^j $, we have

\vspace{1ex}
\noindent
 $ {\rm res}_{\{O\}}( h(x,y), \tau_F dx \wedge dy) = c_{0,0}b_{0,0} + c_{0,1}b_{0,1} + c_{0,2}b_{0,2}
+c_{1,0}b_{1,0} + c_{0,3}b_{0,3}+c_{1,1}b_{1,1}  + c_{0,4}b_{0,4}  
+ c_{1,2}b_{1,2}  + (c_{0,5}-\frac{1}{3}c_{2,0})b_{0,5} + c_{1,3}b_{1,3} 
+ (c_{1,4} -\frac{5}{7}c_{0,6}+\frac{5}{21}c_{2,1})b_{1,4} 
+ (c_{1,5} -\frac{1}{3}c_{3,0} -\frac{5}{7}c_{0,7}+\frac{5}{21}c_{2,2})b_{1,5}. $\end{exmp}

\section{ $ \mu $-constant deformation }

In this section, we consider a $ \mu$-constant deformation of a quasi homogeneous singularity, 
a family of semi-quasi homogeneous isolated hypersurface singularities \citep{gl,lr}. 
We give, as an application 
of the algorithms presented in the previous section, an algorithm for computing Grothendieck 
point residues associated to a $ \mu$-constant deformation of a quasi homogeneous isolated 
hypersurface singularity. The keys of the resulting algorithm are the use of parametric local cohomology
 systems and parametric Gr\"obner systems (comprehensive Gr\"obner systems).

Let $ w =(w_1,w_2,\ldots,w_n) \in {\mathbb N}^n $ be a weight vector for $ z =(z_1, z_2,\ldots, z_n). $
Let $ d_{w}(z^{\lambda}) $ denote the weighted degree of a monomial 
$ z^{\lambda}=z_1^{\ell_1}z_2^{\ell_2} \cdots z_n^{\ell_n} $ defined to be 
$$ d_{w}(z^{\lambda}) = \ell_1 w_1 + \ell_2 w_2 + \cdots + \ell_n w_n. $$

\begin{defn}

\begin{enumerate}
\item[(1)] A non-zero polynomial $ f_0 $ is called a weighted homogeneous (or  quasi homogeneous) polynomial of type  $(d, w)$, if all monomials of $ f_0 $ 
 have the same weighted degree $ d $ with respect to the weight vector $ w $, that is 
$ f_0 = \sum_ {d_w(z^{\lambda}) = d}  c_{\lambda} z^{\lambda}$ where $c_\lambda \in K$. 

\item[(2)] A polynomial $ f(z) = f_0(z) + g(z) $ is called a semi weighted homogeneous (or semi quasi homogeneous) polynomial of type $(d,w)$, if
\item[(i)] $ f_0 $ is weighted homogeneous of type $(d, w)$, and $ f_0(z) =0 $ has an isolated singularity at the origin $ O$,
and 
\item[(ii)] $ g(z) = \sum_{d_w(z^{\beta_j}) > d} b_j z^{\beta_j}, $ where $ b_{j} $ are coefficients.
\end{enumerate}
\end{defn}

Let $ t=(t_1, t_2,\ldots, t_m) $ denote a set of new indeterminates, and let $ T = \{ t \mid t \in {\mathbb C}^m \}. $
Let $$ f_t(z) = f_0(z) +  g(z, t), \  \mbox{with} \ \ g(z,t) = \sum_{d_w(z^{\beta_j}) > d} t_j z^{\beta_j} $$
be a family of semi weighted homogeneous polynomials in $ K(t)[z] $, where $ t \in T $ is regarded as a deformation parameter. Then $ f_t $ is a $ \mu$-constant deformation of $ f_0.$

Set $ F = [ \frac{\partial f}{\partial z_1},  \frac{\partial f}{\partial z_2},\ldots,  \frac{\partial f}{\partial z_n} ]. $
Let $ I_F $ denote a family of ideals in $ K(t)\{ z \} $ generated by $ F $ with the parameter $ t \in T $ and let
$$ H_F = \SetDef{ \psi \in H_{\{O \} }^n (K(t)[z])}{\frac{\partial f}{\partial z_1} \ast \psi = \frac{\partial f}{\partial z_2}\ast \psi = \cdots =
\frac{\partial f}{\partial z_n}\ast \psi = 0}. $$

Let $ \succ $ be a term ordering on $ K(t)[\xi] = K(t)[\xi_1, \xi_2,\ldots,\xi_n] $ compatible with the weight 
vector $ w. $ 
It is known, for semi weighted homogeneous cases, that the set of leading exponents $ \Lambda_F $ is independent 
of $ t $ and thus so is the corresponding basis monomial set $ M_F. $ 
In our previous papers \citep{nt1,nt3}, an algorithm for computing a basis $ \Psi_F $ of $ H_F $ is given.
The algorithm also computes Grothendieck local duality as in the non parametric cases.
The other steps, from step 3 to step 10 in the algorithm {\bf  tau } are also executable by using parametric Gr\"obner systems. The step 1 and step 2 of the algorithm {\bf residues} are also executable. 

Here we give an example of computation.

\begin{exmp}[$ E_{12} $ singularity] Let us consider $ f=x^3+y^7+txy^5 $ ( $ t \ne 0 $).\\
\noindent
{\bf step 1:} \ A basis $ \Psi_F$ of the vector space $ H_F $ with respect to a term ordering $ \succ $ compatible with the weight $ w=(7,3) $ is
\begin{flushleft}
$\left\{1, \eta, \eta^2, \xi, \eta^3, \xi\eta, \eta^4, \xi\eta^2,
\eta^5-\frac{t}{3}\xi^2, \xi\eta^3, 
\xi\eta^4-\frac{5t}{7}\eta^6+\frac{5t^2}{21}\xi^2\eta, \right.$ \\
$\left. \xi\eta^5-\frac{t}{3}\xi^3-\frac{5t}{7}\eta^7+\frac{5t^2}{21}\xi^2\eta^2\right\}.$
\end{flushleft}

\noindent
The set $ \Lambda_F $ is 
\begin{flushleft}
$\Lambda_F = \{ (0,0), (0,1),(0,2), (1,0),(0.3), (1,1),  (0,4), (1,2), (0,5), (1,3), (1,4),$ \\
$ (1,5) \}. $
\end{flushleft}

\noindent
{\bf step 2:} $ x^4, y^8 \in I_F. $

\noindent
{\bf step 5:} $ q(x, y) = 147+25t^3 y \in J_F : J_{F,O}. $

\noindent
{\bf step 6:}  
$
\left(
\begin{array}{c}
q(x,y)x^4\\
q(x,y)y^8
\end{array} 
\right)$ \\
\hspace{1.5cm}$=\left(
\begin{array}{cc}
 (25/3t^3y+49)x^2-49/3ty^5, &  -5/3t^3y^2x+7/3t^2y^4 \\  
 25t^2y^4   & -15tx+21y^2
\end{array} 
\right)\left(
\begin{array}{c}
\frac{\partial f}{\partial x}\\
\frac{\partial f}{\partial y}
\end{array} 
\right).
$


\vspace{1ex}
\noindent
{\bf step 7:} \ $ {\rm Det}(x,y) $ is 

\vspace{1ex}
\noindent
$
(-125t^4y-735t)x^3+(175t^3y^3+1029y^2)x^2+(125/3t^5y^6+245t^2y^5)x-175/3t^4y^8-343ty^7.
$

\vspace{1ex}
\noindent
{\bf step 8:} \ A Gr\"obner basis of $ \langle x^4, y^8, 1-q(x,y)u \rangle $ is 

\vspace{1ex}
\noindent 
$
\{y^8, x^4, -6103515625t^{21}y^7+35888671875t^{18}y^6-211025390625t^{15}y^5 
+12408292$ 
$96875t^{12}y^4-7296076265625t^9y^3+42900928441875t^6y^2 
-252257459238225t^3y-218041257467152161u+1483273860320763\}.
$

\vspace{1ex}
\noindent
{\bf step 9:} \ We have 

\vspace{1ex}
\noindent
$ {\rm Den} = (218041257467152161)^2, $

\vspace{1ex}
\noindent
$
poly(x,y) =-6103515625t^{21}y^7+35888671875t^{18}y^6-211025390625t^{15}y^5+1$ 
$240829296875t^{12}y^4 -7296076265625t^9y^3+42900928441875t^6y^2-252257459238$ 
$225t^3y+1483273860320763,$

\vspace{2ex}
\noindent
$ {\rm NU} = 
-72425481460974755859375000t^{21}y^7+372629102116715118896484375$ 
$t^{18}y^6-1878050674668244199238281250t^{15}y^5+920244830587439657626757812$
$5t^{12}y^4-43288316830833161494762687500t^9y^3+19090147722397424219190345$
$1875t^6y^2-748333790717979029392261531350t^3y+220010134471085834641324$
$8902169,
$

\vspace{1ex}
\noindent
and 

\vspace{1ex}
\noindent
$ 
{\rm Num}=(6654091109227055694580078125t^{22}y^7-39126055722255087484130859$
$375t^{19}y^6+230061207646859914406689453125t^{16}y^5-13527599009635362967113$
$33984375t^{13}y^4+7954228217665593424662643828125t^{10}y^3-46770861919873689$
$337016345709375t^7y^2+27501266808885729330165611277112tt^4y-16170744883$
$62480884613737943094215t)x^3+(-322085690705603880169365234375t^{15}y^7+1$
$893863861348950815395867578125t^{12}y^6-11135919504731830794527701359375$
$t^9y^5+65479206687823165071822883993125t^6y^4-38501773532440021062231855$
$7879575t^3y^3+2263904283707473238459233120331901y^2)x^2+(155902873066245$
$63112338781903125t^8y^7-91670889362952431100552037590375t^5y^6+539024829$ 
$454160294871245981031405t^2y^5)x-754634761235824412819744373443967ty^7.$

\vspace{1ex}

As an output 
we thus have
\begin{center}
$\displaystyle  \tau_F= \sum_{0 \leq i, j \leq 5} b_{i,j}\psi_{i,j}, $
\end{center}
where \\
\noindent
$ b_{0,0} = 30517578125t^{22}/218041257467152161, \\
b_{0,1}=-1220703125t^{19}/1483273860320763, \\ b_{0,2}=48828125t^{16}/10090298369529, b_{0,3} = -1953125t^{13}/68641485507, \\
b_{0,4} = 78125t^{10}/466948881, b_{0,5} = -3125t^{7}/3176523, \\
b_{1,0} = -9765625t^{15}/1441471195647,\\ 
b_{1,1} = 390625t^{12}/9805926501, 
b_{1,2}=-15625t^{9}/66706983, \\
b_{1,3}=625t^{6}/453789, 
b_{1,4}=-25t^{3}/3087,
b_{1,5}=1/21. $ \\
\noindent 
and 

\noindent
$ \psi_{0,0}=1, \ \psi_{0,1}=\eta, \ \psi_{0,2}=\eta^2, \ \psi_{0,3}=\eta^3, \ \psi_{0,4}=\eta^4, \ \psi_{0,5}=\eta^5-\dfrac{t}{3}\xi^2, \\
\psi_{1,0}=\xi, \ \psi_{1,1}=\xi\eta, \ 
\psi_{1,2}=\xi\eta^2, \ \psi_{1,3}=\xi\eta^3, \ 
\psi_{1,4}=\xi\eta^4-\dfrac{5t}{7}\eta^6 +\dfrac{5t^2}{21}\xi^2\eta, \\
\psi_{1,5} = \xi\eta^5-\dfrac{5t}{7}\eta^7 - \dfrac{t}{3}\xi^3 + \dfrac{5t^2}{21}\xi^2\eta^2. $

\vspace{2ex}

We have, for instance, 
\begin{center}
$\displaystyle  {\rm res}_{\{O\}}\left(\frac{dx \wedge dy}{\frac{\partial f}{\partial x} \frac{\partial f}{\partial y}}\right)
= \frac{30517578125}{218041257467152161}t^{22}. $
\end{center}
\end{exmp}




\bibliographystyle{elsarticle-harv}

\bibliography{tajima}  

\begin{thebibliography}{50}
\expandafter\ifx\csname natexlab\endcsname\relax\def\natexlab#1{#1}\fi
\providecommand{\url}[1]{\texttt{#1}}
\providecommand{\href}[2]{#2}
\providecommand{\path}[1]{#1}
\providecommand{\DOIprefix}{doi:}
\providecommand{\ArXivprefix}{arXiv:}
\providecommand{\URLprefix}{URL: }
\providecommand{\Pubmedprefix}{pmid:}
\providecommand{\doi}[1]{\href{http://dx.doi.org/#1}{\path{#1}}}
\providecommand{\Pubmed}[1]{\href{pmid:#1}{\path{#1}}}
\providecommand{\bibinfo}[2]{#2}
\ifx\xfnm\relax \def\xfnm[#1]{\unskip,\space#1}\fi
\bibitem[{Altman and Kleiman(1970)}]{ak}
\bibinfo{author}{Altman, A.}, \bibinfo{author}{Kleiman, S.},
  \bibinfo{year}{1970}.
\newblock \bibinfo{title}{Introduction to {G}rothendieck {D}uality {T}heory,
  Lecture Notes in Math. \textbf{146}}.
\newblock \bibinfo{publisher}{Springer}.
\bibitem[{B{\u a}nic{\u a} and St{\u a}n{\u a}{\c s}il{\u a}(1974)}]{bs}
\bibinfo{author}{B{\u a}nic{\u a}, C.}, \bibinfo{author}{St{\u a}n{\u a}{\c
  s}il{\u a}, O.}, \bibinfo{year}{1974}.
\newblock \bibinfo{title}{M{\'e}thodes {A}lgebriques dans la {T}h{\'e}orie
  {G}lobale des {E}spaces {C}omplexes}.
\newblock \bibinfo{publisher}{Gauthier-Villars}.
\bibitem[{Baum and Bott(1972)}]{bb2}
\bibinfo{author}{Baum, P.F.}, \bibinfo{author}{Bott, R.}, \bibinfo{year}{1972}.
\newblock \bibinfo{title}{Singularities of holomorphic foliations}.
\newblock \bibinfo{journal}{J. {D}ifferential {G}eometry} \bibinfo{volume}{7},
  \bibinfo{pages}{279--342}.
\bibitem[{Boyer and Hickel(1997)}]{bh}
\bibinfo{author}{Boyer, J.Y.}, \bibinfo{author}{Hickel, M.},
  \bibinfo{year}{1997}.
\newblock \bibinfo{title}{Une g\'en\'eralisation de la loi de transformation
  pour les r\'esidus}.
\newblock \bibinfo{journal}{Bull. Soc. Math. France} \bibinfo{volume}{125},
  \bibinfo{pages}{315--335}.
\bibitem[{Brasselet et~al.(2009)Brasselet, Seade and Suwa}]{bss}
\bibinfo{author}{Brasselet, J.P.}, \bibinfo{author}{Seade, J.},
  \bibinfo{author}{Suwa, T.}, \bibinfo{year}{2009}.
\newblock \bibinfo{title}{Vector {F}ields on {S}ingular {V}arieties, {L}ecture
  {N}otes in {M}ath. \textbf{1987}}.
\newblock \bibinfo{publisher}{Springer}.
\bibitem[{Bykov et~al.(1991)Bykov, Kytmanov, Lazman and Passare}]{bklp}
\bibinfo{author}{Bykov, V.}, \bibinfo{author}{Kytmanov, A.},
  \bibinfo{author}{Lazman, M.}, \bibinfo{author}{Passare, M.},
  \bibinfo{year}{1991}.
\newblock \bibinfo{title}{Elimination {M}ethods in {P}olynomial {C}omputer
  {A}lgebra, {M}athematics and its {A}pplications}.
\newblock \bibinfo{publisher}{Kluwer}.
\bibitem[{Cardinal and Mourrain(1996)}]{cm}
\bibinfo{author}{Cardinal, J.P.}, \bibinfo{author}{Mourrain, B.},
  \bibinfo{year}{1996}.
\newblock \bibinfo{title}{Algebraic approach of residues and applications}, in:
  \bibinfo{booktitle}{The Mathematics of Numerical Analysis, Lectures in
  Applied Math. {\bf 32}}, \bibinfo{publisher}{AMS}. pp.
  \bibinfo{pages}{186--210}.
\bibitem[{Cattani et~al.(1996)Cattani, Dickenstein and Sturmfels}]{cds}
\bibinfo{author}{Cattani, E.}, \bibinfo{author}{Dickenstein, A.},
  \bibinfo{author}{Sturmfels, B.}, \bibinfo{year}{1996}.
\newblock \bibinfo{title}{Computing multidimensional residues}.
\newblock \bibinfo{journal}{Progress in Math.} \bibinfo{volume}{143},
  \bibinfo{pages}{135--164}.
\bibitem[{Cherveny(2018)}]{c}
\bibinfo{author}{Cherveny, L.}, \bibinfo{year}{2018}.
\newblock \bibinfo{title}{Remarks on localizing {F}utaki-{M}orita integrals at
  isolated degenerate zeros}.
\newblock \bibinfo{journal}{Differential Geom. Appl.} \bibinfo{volume}{56},
  \bibinfo{pages}{1--12}.
\bibitem[{Corr\^ea et~al.(2016)Corr\^ea, Rodriguez and Soares}]{co}
\bibinfo{author}{Corr\^ea, M.}, \bibinfo{author}{Rodriguez, P.},
  \bibinfo{author}{Soares, M.G.}, \bibinfo{year}{2016}.
\newblock \bibinfo{title}{A {B}ott-type residue formula on complex orbifolds}.
\newblock \bibinfo{journal}{Int. Math. Res. Not. IMRN} \bibinfo{volume}{10},
  \bibinfo{pages}{2889--2911}.
\bibitem[{Dickenstein and Sessa(1991)}]{ds}
\bibinfo{author}{Dickenstein, A.E.}, \bibinfo{author}{Sessa, C.},
  \bibinfo{year}{1991}.
\newblock \bibinfo{title}{Duality methods for the membership problem}.
\newblock \bibinfo{journal}{Progress in Math.} \bibinfo{volume}{94},
  \bibinfo{pages}{86--103}.
\bibitem[{Elkadi and Mourrain(2007)}]{em2}
\bibinfo{author}{Elkadi, M.}, \bibinfo{author}{Mourrain, B.},
  \bibinfo{year}{2007}.
\newblock \bibinfo{title}{Introduction \`a la {R}\'esolution des {S}yst\`emes
  {P}olynomiaux}.
\newblock \bibinfo{publisher}{Springer}.
\bibitem[{Greuel(1986)}]{gl}
\bibinfo{author}{Greuel, G.M.}, \bibinfo{year}{1986}.
\newblock \bibinfo{title}{Constant {M}ilnor number implies constant
  multiplicity for quasi homogeneous singularities}.
\newblock \bibinfo{journal}{Manuscripta Math.} \bibinfo{volume}{56},
  \bibinfo{pages}{156--166}.
\bibitem[{Griffiths(1976)}]{g}
\bibinfo{author}{Griffiths, P.}, \bibinfo{year}{1976}.
\newblock \bibinfo{title}{Variations on a theorem of {A}bel}.
\newblock \bibinfo{journal}{Invent. Math.} \bibinfo{volume}{35},
  \bibinfo{pages}{321--390}.
\bibitem[{Griffiths and Harris(1978)}]{gh}
\bibinfo{author}{Griffiths, P.}, \bibinfo{author}{Harris, J.},
  \bibinfo{year}{1978}.
\newblock \bibinfo{title}{Principles of {A}lgebraic {G}eometry}.
\newblock \bibinfo{publisher}{Wiley Interscience}.
\bibitem[{Grothendieck(1957)}]{gr}
\bibinfo{author}{Grothendieck, A.}, \bibinfo{year}{1957}.
\newblock \bibinfo{title}{Th\'eor\`emes de dualit\'e pour les faisceaux
  alg\'ebriques coh\'erents. S\'eminaire Bourbaki {\bf 149}}.
\bibitem[{Grothendieck(1967)}]{hartshorne1}
\bibinfo{author}{Grothendieck, A.}, \bibinfo{year}{1967}.
\newblock \bibinfo{title}{{L}ocal {C}ohomology, notes by {R}. {H}artshorne,
  Lecture Notes in Math., \text{41}}.
\newblock \bibinfo{publisher}{Springer}.
\bibitem[{Hartshorne(1966)}]{hartshorne}
\bibinfo{author}{Hartshorne, R.}, \bibinfo{year}{1966}.
\newblock \bibinfo{title}{Residues and Duality. Lecture Notes in Math. {\bf
  20}}.
\newblock \bibinfo{publisher}{Springer}.
\bibitem[{Klehn(2002)}]{kl}
\bibinfo{author}{Klehn, O.}, \bibinfo{year}{2002}.
\newblock \bibinfo{title}{Local residues of holomorphic 1-forms on an isolated
  surface singularity}.
\newblock \bibinfo{journal}{Manuscripta Math.} \bibinfo{volume}{109},
  \bibinfo{pages}{93--108}.
\bibitem[{Kulikov(1998)}]{kul}
\bibinfo{author}{Kulikov, V.S.}, \bibinfo{year}{1998}.
\newblock \bibinfo{title}{Mixed {H}odge {S}tructure and {S}ingularities,
  {C}ambridge {T}racts in {M}ath. {\bf 132}}.
\newblock \bibinfo{publisher}{Cambridge Univ. Press}.
\bibitem[{Kunz(2009)}]{ku}
\bibinfo{author}{Kunz, E.}, \bibinfo{year}{2009}.
\newblock \bibinfo{title}{Residues and {D}uality for {P}rojective {A}lgebraic
  {V}arieties. Univ. Lecture Series {\bf 47}}.
\newblock \bibinfo{publisher}{AMS}.
\bibitem[{Kytmanov(1988)}]{ky}
\bibinfo{author}{Kytmanov, A.M.}, \bibinfo{year}{1988}.
\newblock \bibinfo{title}{A transformation formula for {G}rothendieck residues
  and some of its applications}.
\newblock \bibinfo{journal}{Siberian Math. J.} \bibinfo{volume}{169},
  \bibinfo{pages}{495--499}.
\bibitem[{L\^e and Ramanujam(1976)}]{lr}
\bibinfo{author}{L\^e, D.T.}, \bibinfo{author}{Ramanujam, C.P.},
  \bibinfo{year}{1976}.
\newblock \bibinfo{title}{The invariance of {M}ilnor's number implies the
  imvariance of the topological type}.
\newblock \bibinfo{journal}{Amer. J. Math.} \bibinfo{volume}{98},
  \bibinfo{pages}{67--78}.
\bibitem[{Lehmann(1991)}]{le}
\bibinfo{author}{Lehmann, D.}, \bibinfo{year}{1991}.
\newblock \bibinfo{title}{R\'esidues des sous-vari\'etes invariantes d'un
  feuilletage singulier}.
\newblock \bibinfo{journal}{Ann. Inst. Fourier, Grenoble} \bibinfo{volume}{41},
  \bibinfo{pages}{211--258}.
\bibitem[{Lipman(1984)}]{l}
\bibinfo{author}{Lipman, J.}, \bibinfo{year}{1984}.
\newblock \bibinfo{title}{Dualizing {S}heaves, {D}ifferentials and {R}esidues
  on {A}lgebraic {V}arieties}.
\newblock \bibinfo{journal}{Ast\'erisque} \bibinfo{volume}{117}.
\bibitem[{Lipman(2002)}]{l2}
\bibinfo{author}{Lipman, J.}, \bibinfo{year}{2002}.
\newblock \bibinfo{title}{Lectures on local cohomology and duality}.
\newblock \bibinfo{journal}{Lecture Notes in Pure and Applied Math.}
  \bibinfo{volume}{226}, \bibinfo{pages}{39--89}.
\bibitem[{Mourrain(1997)}]{m}
\bibinfo{author}{Mourrain, B.}, \bibinfo{year}{1997}.
\newblock \bibinfo{title}{Isolated points, duality and residues}.
\newblock \bibinfo{journal}{J. Pure and Applied Algebra}
  \bibinfo{volume}{117/118}, \bibinfo{pages}{460--493}.
\bibitem[{Nabeshima and Tajima(2015a)}]{nt2}
\bibinfo{author}{Nabeshima, K.}, \bibinfo{author}{Tajima, S.},
  \bibinfo{year}{2015}a.
\newblock \bibinfo{title}{Computing logarithmic vector fields associated with
  parametric semi-quasihomogeneous hypersurface isolated singularities}, in:
  \bibinfo{booktitle}{International Symposium on Symbolic and Algebraic
  Computation}, \bibinfo{publisher}{ACM}. pp. \bibinfo{pages}{291--298}.
\bibitem[{Nabeshima and Tajima(2015b)}]{nt3}
\bibinfo{author}{Nabeshima, K.}, \bibinfo{author}{Tajima, S.},
  \bibinfo{year}{2015}b.
\newblock \bibinfo{title}{Efficient computation of algebraic local cohomology
  classes and change of ordering for zero-dimensinal standard bases}, in:
  \bibinfo{booktitle}{International Workshop on Computer Algebra in Scientific
  Computing 2015, Lecture Notes in Computer Science, {\bf 9301}},
  \bibinfo{publisher}{Springer}. pp. \bibinfo{pages}{334--348}.
\bibitem[{Nabeshima and Tajima(2015c)}]{nt1}
\bibinfo{author}{Nabeshima, K.}, \bibinfo{author}{Tajima, S.},
  \bibinfo{year}{2015}c.
\newblock \bibinfo{title}{On the computation of algebraic local cohomology
  classes associated with semi-quasihomogeneous singularities}.
\newblock \bibinfo{journal}{Advanced Studies in Pure Mathematics}
  \bibinfo{volume}{66}, \bibinfo{pages}{143--159}.
\bibitem[{Nabeshima and Tajima(2016a)}]{nt5}
\bibinfo{author}{Nabeshima, K.}, \bibinfo{author}{Tajima, S.},
  \bibinfo{year}{2016}a.
\newblock \bibinfo{title}{Computing {T}jurina stratifications of $\mu$-constant
  deformations via parametric local cohomology systems}.
\newblock \bibinfo{journal}{Applicable Algebra in Engineering, Computation and
  Computing} \bibinfo{volume}{27}, \bibinfo{pages}{451--467}.
\bibitem[{Nabeshima and Tajima(2016b)}]{nt4}
\bibinfo{author}{Nabeshima, K.}, \bibinfo{author}{Tajima, S.},
  \bibinfo{year}{2016}b.
\newblock \bibinfo{title}{Solving extended ideal membership problems in rings
  of convergent power series via {G}r\"obner bases}, in:
  \bibinfo{booktitle}{International Conference on Mathematical Aspects of
  Computer and Information Sciences 2016, Lecture Notes in Computer Science,
  {\bf 9582}}, \bibinfo{publisher}{Springer}. pp. \bibinfo{pages}{252--267}.
\bibitem[{Nabeshima and Tajima(2017)}]{nt6}
\bibinfo{author}{Nabeshima, K.}, \bibinfo{author}{Tajima, S.},
  \bibinfo{year}{2017}.
\newblock \bibinfo{title}{Algebraic local cohomology with parameters and
  parametric standard bases for zero-dimensional ideals}.
\newblock \bibinfo{journal}{Journal of Symbolic Computation}
  \bibinfo{volume}{82}, \bibinfo{pages}{91--122}.
\bibitem[{Noro and Takeshima(1992)}]{noro}
\bibinfo{author}{Noro, M.}, \bibinfo{author}{Takeshima, T.},
  \bibinfo{year}{1992}.
\newblock \bibinfo{title}{Risa/{A}sir- {A} computer algebra system}, in:
  \bibinfo{booktitle}{International Symposium on Symbolic and Algebraic
  Computation}, \bibinfo{publisher}{ACM}. pp. \bibinfo{pages}{387--396}.
\bibitem[{O'Brian(1975)}]{ob}
\bibinfo{author}{O'Brian, N.R.}, \bibinfo{year}{1975}.
\newblock \bibinfo{title}{Zeros of holomorphic vector fields and the
  {G}rothendieck residues}.
\newblock \bibinfo{journal}{Bull. London Math. Soc.} \bibinfo{volume}{7},
  \bibinfo{pages}{33--38}.
\bibitem[{O'Brian(1977)}]{ob2}
\bibinfo{author}{O'Brian, N.R.}, \bibinfo{year}{1977}.
\newblock \bibinfo{title}{Zeros of holomorphic vector fields and {G}rothendieck
  duality theory}.
\newblock \bibinfo{journal}{Trans. AMS.} \bibinfo{volume}{229},
  \bibinfo{pages}{289--306}.
\bibitem[{Ohara and Tajima(2019a)}]{ot1}
\bibinfo{author}{Ohara, K.}, \bibinfo{author}{Tajima, S.},
  \bibinfo{year}{2019}a.
\newblock \bibinfo{title}{An algorithm for computing {G}rothendieck local
  residues {I} - shape base case -}.
\newblock \bibinfo{journal}{Mathematics in Computer Science}
  \bibinfo{volume}{1-2}, \bibinfo{pages}{205--216}.
\bibitem[{Ohara and Tajima(2019b)}]{ot2}
\bibinfo{author}{Ohara, K.}, \bibinfo{author}{Tajima, S.},
  \bibinfo{year}{2019}b.
\newblock \bibinfo{title}{An algorithm for computing {G}rothendieck local
  residues {II} - general case -}.
\newblock \bibinfo{journal}{to apper in Mathematics in Computer Science}
  \bibinfo{note}{ArXiv:1811.08054}.
\bibitem[{Perotti(1998)}]{p}
\bibinfo{author}{Perotti, A.}, \bibinfo{year}{1998}.
\newblock \bibinfo{title}{Multidimensional residues and ideal membership}.
\newblock \bibinfo{journal}{Publ. Mat.} \bibinfo{volume}{42},
  \bibinfo{pages}{143--152}.
\bibitem[{Saito(1983)}]{sai}
\bibinfo{author}{Saito, S.}, \bibinfo{year}{1983}.
\newblock \bibinfo{title}{The higher residue pairings for a family of
  hypersurface singular points}, in: \bibinfo{booktitle}{Symposia in Pure Math.
  {\bf 40}, Part 2}, pp. \bibinfo{pages}{441--463}.
\bibitem[{Sato and Suzuki(2009)}]{ss}
\bibinfo{author}{Sato, Y.}, \bibinfo{author}{Suzuki, A.}, \bibinfo{year}{2009}.
\newblock \bibinfo{title}{Computation of inverses in residue class rings of
  parametric polynomial ideals}, in: \bibinfo{booktitle}{International
  Symposium on Symbolic and Algebraic Computation}, \bibinfo{publisher}{ACM}.
  pp. \bibinfo{pages}{311--315}.
\bibitem[{Suwa(1988)}]{s}
\bibinfo{author}{Suwa, T.}, \bibinfo{year}{1988}.
\newblock \bibinfo{title}{Indices of {V}ector {F}ields and {R}esidues of
  {S}ingular {H}olomorphic {F}oliations}.
\newblock \bibinfo{publisher}{Hermann}.
\bibitem[{Suwa(2005)}]{s2}
\bibinfo{author}{Suwa, T.}, \bibinfo{year}{2005}.
\newblock \bibinfo{title}{Residues of {C}hern classes on singular varieties}.
\newblock \bibinfo{journal}{{S}\'eminaires et {C}ongr\`es, Soc. Math. France.}
  \bibinfo{volume}{10}, \bibinfo{pages}{265--285}.
\bibitem[{Tajima and Nakamura(2005a)}]{tn1}
\bibinfo{author}{Tajima, S.}, \bibinfo{author}{Nakamura, Y.},
  \bibinfo{year}{2005}a.
\newblock \bibinfo{title}{Algebraic local cohomology classes attached to
  quasi-homogeneous hypersurface isolated singularities}.
\newblock \bibinfo{journal}{Publ. Res. Inst. Math. Sci.} \bibinfo{volume}{41},
  \bibinfo{pages}{1--10}.
\bibitem[{Tajima and Nakamura(2005b)}]{tn2}
\bibinfo{author}{Tajima, S.}, \bibinfo{author}{Nakamura, Y.},
  \bibinfo{year}{2005}b.
\newblock \bibinfo{title}{Computational aspects of {G}rothendieck local
  residues}.
\newblock \bibinfo{journal}{S\'eminaires et Congr\`es, oc. Math. France.}
  \bibinfo{volume}{10}, \bibinfo{pages}{287--305}.
\bibitem[{Tajima and Nakamura(2009)}]{tn3}
\bibinfo{author}{Tajima, S.}, \bibinfo{author}{Nakamura, Y.},
  \bibinfo{year}{2009}.
\newblock \bibinfo{title}{Annihilating ideals for an algebraic local cohomology
  class}.
\newblock \bibinfo{journal}{Journal of Symbolic Computation}
  \bibinfo{volume}{44}, \bibinfo{pages}{435--448}.
\bibitem[{Tajima et~al.(2009)Tajima, Nakamura and Nabeshima}]{tnn}
\bibinfo{author}{Tajima, S.}, \bibinfo{author}{Nakamura, Y.},
  \bibinfo{author}{Nabeshima, K.}, \bibinfo{year}{2009}.
\newblock \bibinfo{title}{Standard bases and algebraic local cohomology for
  zero dimensional ideals}.
\newblock \bibinfo{journal}{Advanced Studies in Pure Mathematics}
  \bibinfo{volume}{56}, \bibinfo{pages}{341--361}.
\bibitem[{Tong(1973)}]{ton}
\bibinfo{author}{Tong, Y.L.}, \bibinfo{year}{1973}.
\newblock \bibinfo{title}{Integral representation formulae and {G}rothendieck
  residue symbol}.
\newblock \bibinfo{journal}{Amer. J. Math.} \bibinfo{volume}{95},
  \bibinfo{pages}{904--917}.
\bibitem[{Varchenko(1986)}]{v}
\bibinfo{author}{Varchenko, A.N.}, \bibinfo{year}{1986}.
\newblock \bibinfo{title}{On the local residue and the intersection form on the
  vanishing cohomology}.
\newblock \bibinfo{journal}{Math. USSR Izvestiya} \bibinfo{volume}{26},
  \bibinfo{pages}{31--52}.
\bibitem[{Yushakov(1984)}]{y}
\bibinfo{author}{Yushakov, A.P.}, \bibinfo{year}{1984}.
\newblock \bibinfo{title}{On the computation of the complete sum of residues
  relative to a polynomial mapping in $ {\mathbb c}^n $}.
\newblock \bibinfo{journal}{Akad. Nauk. SSSR} \bibinfo{volume}{275},
  \bibinfo{pages}{817--820}.

\end{thebibliography}








\end{document}